 \documentclass[final,5p,times,twocolumn,numeric]{elsarticle}

 \usepackage{epsfig}
\usepackage{natbib}
\usepackage{amsmath}
\usepackage{graphicx}
\usepackage{amsfonts}
\usepackage{latexsym}
\usepackage{bbold}
\usepackage{wasysym}
\usepackage{calligra}
\usepackage{float}
\usepackage{ulem}
\usepackage{inputenc}
\usepackage{xspace}
\usepackage{url}
\usepackage{epstopdf}
\usepackage{tikz}
\usepackage{amssymb}
\usepackage{enumitem}   
\setcitestyle{square}

\newcommand\Lie{\pounds}

\usepackage{appendix}

\usepackage{amsmath}
\usepackage{amssymb}
\usepackage{graphicx}
\usepackage{amsfonts}
\usepackage{latexsym}
\usepackage{bbold}
\usepackage{wasysym}
\usepackage{calligra}
\usepackage{float}
\usepackage{ulem}
\usepackage{inputenc}
\usepackage{xspace}
\usepackage{url}
\usepackage{epstopdf}
\usepackage{tikz}
\usepackage{amsthm}
\usepackage{soul}

\newtheorem*{mydef*}{Definition}

\newtheorem*{theorem*}{Theorem}
\newtheorem{Identity}{Identity}

\newcommand{\be}{\begin{equation}}
\newcommand{\ee}{\end{equation}}
\newcommand{\beq} {\begin{equation}}
\newcommand{\eeq} {\end{equation}}
\newcommand{\ba}{\begin{eqnarray}}
\newcommand{\ea}{\end{eqnarray}}

\usepackage{amssymb}
\usepackage{lipsum}

\begin{document}

\begin{frontmatter}



\title{Lagrangian Dynamics of Spinning Pole-Dipole-Quadrupole    Particles in
Metric-Affine Geometries}


\author[first,f]{Damianos Iosifidis}
\affiliation[first]{organization={Scuola Superiore Meridionale, Largo San Marcellino 10, 80138 Napoli, Italy},
}

\affiliation[f]{organization={INFN– Sezione di Napoli, Via Cintia, 80126 Napoli, Italy}
}

\begin{abstract}
We construct the Lagrangian formulation of a micro-structured spinning, dilating and shearing (deformable) test body, moving in arbitrary non-Riemannian backgrounds possessing all geometrical entities of curvature, torsion and non-metricity. We start with a Lagrangian of a generic form that depends on the particle's velocity, its material frame and its absolute derivative, and the background geometry consisting  of a metric and an independent affine connection. Performing variations of the path and the material frame, we derive the equations of motion for the particle that govern the evolution of its momentum and hypermomentum in this generic background. The reported equations of motion generalize those of a spinning particle (Mathisson \cite{Mathisson:1937zz}, Papapetrou \cite{Papapetrou:1951pa}, Dixon \cite{Dixon:1974xoz}) by the inclusion of the dilation and shear (hadronic) currents of matter. Using the derived equations of motion, a generalized conserved quantity is also found. Further conserved quantities that can be obtained by appropriate supplementary conditions are also discussed.

\end{abstract}




\end{frontmatter}




\section{Introduction}
\label{introduction}

The  motion of a spinning test body under the influence of the gravitational field is  a fundamental theoretical problem of ultimate importance. The equations of motion describing such a spinning body in a given Riemannian background have been available for quite some time due to the pioneering works of Mathisson \cite{Mathisson:1937zz}, Papapetrou, \cite{Papapetrou:1951pa} and later Dixon \cite{Dixon:1974xoz}. The main conclusion from these works is that a spinning body does not follow a geodesic; there are additional forces pulling it away from the geodesic trajectory. The Lagrangian formulation for the equations of motion of a spinning particle was subsequently developed in \cite{Bailey:1975fe}. Restricting the variable dependence of the Lagrangian in a certain way, it is then possible to describe spinning tops, as has been done in \cite{Hanson:1974qy} for Minkowski space and later generalized for Riemannian spaces in \cite{hojman1976spinning}. 

These studies focus on motion over flat Minkowski or Riemannian backgrounds. It is then interesting to find out what kind of particles need to be used in order to measure possible departures from the Riemannian geometry and be able to probe the non-Riemannian (i.e. torsion and non-metricity) aspects of the spacetime structure. As it turns out, the way to detect these possible meta-Riemannian effects of spacetime, is by using particles with internal structure as probes \cite{Puetzfeld:2007ye,Puetzfeld:2007hr}. In general, these particles contain all three pieces of the  hypermomentum current \cite{hehl1976hypermomentum}, namely those of spin (antisymmetric), dilation (trace) and shear (symmetric traceless).

The equations of motion governing the trajectory that a microstructured test body follows in a generic non-Riemannian background  have been derived in \cite{obukhov2014conservation} using the conservation laws of Metric-Affine Gravity and the multipole moments method. Another approach was used in \cite{Iosifidis:2023eom} where instead the convecting $ansatz \ddot{e}$ of the energy tensors were used along with the conservation laws. The results obtained by these two seemingly different approaches are in perfect agreement with one another. However, in both of these works, no Lagrangian was used to derive the associated equations of motion from a variational principle.\footnote{Of course, it is quite essential to ask wether these path equations can be derived at all from a Lagrangian. The content of this paper proves that, indeed, such trajectories do follow from a variational principle.} The Lagrangian formulation of the motion offers many advantages. For instance, having a Lagrangian, one unambiguously defines the canonical momentum as well as the rest of the important kinematic variables. Furthermore, in order to develop singularity theorems, the existence of a Lagrangian is unavoidable. It is therefore important to have a Lagrangian description for the motion of test bodies (or particles) in external backgrounds. Generalizing the purely Riemannian only spin case developed in \cite{Bailey:1975fe}, in this work, we extend these results by considering a general Metric-Affine space and a body possessing also the dilation and shear charges of hypermomentum (along with the spin). 

The paper is organized as follows: We first introduce the underlying geometric structure by defining the associated important geometric entities. In this background we define the   relevant kinematic variables. We subsequently construct the Lagrangian to be varied upon. Performing arbitrary path and material frame variations we derive the equations of motion for the test body. Finally, using the derived equations and the notion of a generalized Killing vector field we explicitly obtain a conserved quantity of the motion.

\section{Geometry and Dynamics}
We consider a 4-dimensional non-Riemannian space consisting of a metric $g$ and an independent  linear connection $\nabla$, with components in local coordinates $g_{\mu\nu}$ and $\Gamma^{\lambda}{}_{\mu\nu}$, respectively. The torsion and curvature of the connection are defined by
\beq
S_{\mu\nu}^{\;\;\;\lambda}:=\Gamma^{\lambda}_{\;\;\;[\mu\nu]}
\eeq
and 
\beq
R^{\mu}_{\;\;\;\nu\alpha\beta}:= 2\partial_{[\alpha}\Gamma^{\mu}_{\;\;\;|\nu|\beta]}+2\Gamma^{\mu}_{\;\;\;\rho[\alpha}\Gamma^{\rho}_{\;\;\;|\nu|\beta]} \label{R}
\eeq
respectively. The non-compatibility of the generic connection is quantified in the non-metricity tensor
\beq
Q_{\alpha\mu\nu}:=- \nabla_{\alpha}g_{\mu\nu}
\eeq
The above three geometric entities fully characterize this generalized non-Riemannian space. More compactly, the deviation from the torsionless and metric-compatible Riemannian geometry is encoded in the difference between the general connection and the one of Levi-Civita. In particular, if we denote by $\widetilde{\Gamma}^{\lambda}{}_{\mu\nu}$ the Levi-Civita connection, then the difference 
\beq
N^{\lambda}{}_{\mu\nu}=\Gamma^{\lambda}{}_{\mu\nu}-\widetilde{\Gamma}^{\lambda}{}_{\mu\nu} \label{defN}
\eeq
defines the so-called distortion tensor \cite{hehl1995metric}. The torsion and non-metricity are easily retrieved through the relations \cite{hehl1995metric,iosifidis2019metric}
\beq
S_{\mu\nu\alpha}=N_{\alpha[\mu\nu]}\;\;,\;\;\; Q_{\nu\alpha\mu}=2 N_{(\alpha\mu)\nu} \label{QNSN}
\eeq
In addition, with the use of (\ref{defN}) each quantity can be separated into its Riemannian part, denoted by an overhead tilde, and its post-Riemannian contributions of torsion and non-metricity.

Moving on to the kinematics, let $x^{\mu}(\lambda)$ be the  representative   worldline of an extended test body (or particle) in terms of the arbitrary parameter $\lambda$. The un-normalized velocity field is then defined by
\beq
\upsilon^{\mu}:=\frac{d x^{\mu}}{d \lambda}
\eeq
If $\tau$ is the proper time, then the normailized velocity field is given as usual
\beq
u^{\mu}:=\frac{d x^{\mu}}{d \tau}\;\;, \;\;u_{\mu}u^{\mu}=-1
\eeq
Furthermore, the internal structure of the particle (i.e. spin, dilation and shear) is described by the material frame (or tetrad or vierbein) field  $e_{\mu}{}^{a}(\lambda)=e_{\mu}{}^{(a)}(\lambda)$ with $a=0,1,2,3$, attached to each point of the representative curve. Generally, this material frame is 'elastic' in the sense that it undergoes arbitrary deformations upon transportation in space.    

The absolute derivative with respect to the affine parameter along the world-line will be denoted as $\frac{\delta}{\delta \lambda}\equiv \frac{D}{d \lambda}$. In particular, when acting on the frame field: 
\beq
\dot{e}_{\mu}{}^{a}:=\frac{D e_{\mu}{}^{a}}{d \lambda}=\frac{d e_{\mu}{}^{a}}{d \lambda}-\Gamma^{\lambda}{}_{\mu\nu}e_{\lambda}{}^{a}\upsilon^{\nu} \label{abscove}
\eeq

Extending the results of \cite{Bailey:1975fe} to Metric-Affine Geometries and particles with microstructure, we now consider a generic Lagrangian of the form $L(\upsilon^{\mu}, e_{\mu}{}^{a}, \dot{e}_{\mu}{}^{a}, g_{\mu\nu})$ with the corresponding reparametrization invariant particle action\footnote{Similar to \cite{Bailey:1975fe}, we only specify the geometric (background) and kinematic variables that this Lagrangian depends upon, without specifying its exact form. This approach is also reinforced by the fact that no specific form of a Lagrangian should be postulated  to derive the path equations.}
\beq
I=\int  L(\upsilon^{\mu}, e_{\mu}{}^{a}, \dot{e}_{\mu}{}^{a}, g_{\mu\nu}) d \lambda \label{action}
\eeq
Let us note that there is also an implicit dependence on the affine-connection hidden in $\dot{e}_{\mu}{}^{a}$ as is obvious from (\ref{abscove}). We should also stress that the material tetrad $e_{\mu}{}^{a}$ has a priori nothing to do with the geometry tetrad (vierbein) $h_{\mu}{}^{a}$ which defines the metric through $g_{\mu\nu}=h_{\mu}{}^{a}h_{\nu}{}^{b}\eta_{ab}$. However one can always use the freedom in choosing $h_{\mu}{}^{a}$ and make the latter coincide with the material frame at the cost of loosing the orthonormality condition on it. The inverse co-frame $e^{\mu}{}_{b}$ of the material frame field is defined through $e_{\mu}{}^{a}e^{\mu}{}_{b}=\delta^{a}_{b}$.  In general the material frame field is not orthonormal in the sense that the induced internal metric $g_{ab}=e^{\mu}{}_{a}e^{\nu}{}_{b}g_{\mu\nu}$ is not the flat Minkowski metric. This 'elasticity' of the material frame fields is due to spacetime non-metricity and also reflects the fact that there is a genuine number of $16$ degrees of freedom encoded in $e_{\mu}{}^{a}$ instead of the $6$ ones that one has in the Riemannian case, see e.g. \cite{Bailey:1975fe}.

The canonical momentum is defined through
\beq
P_{\mu}:=\frac{\partial L}{\partial \upsilon^{\mu}}
\eeq
and here we also define the excitations of hypermomentum and energy-momentum according to
\beq
H_{\nu}{}^{\mu}:=2 e_{\nu}{}^{a} \frac{\partial L}{\partial \dot{e}_{\mu}{}^{a}}\;\;, \;\; t^{\mu\nu}:=2 \frac{\partial L}{\partial g_{\mu\nu}}
\eeq
The antisymmetric part of the former gives the spin tensor, i.e. $S_{\mu\nu}=H_{[\mu\nu]}$ (compare also with \cite{Bailey:1975fe}). In our construction, the hypermomentum excitation is generic and  has also a trace (dilation) as well as a symmetric trace free part (shear), in order to include all microproperties of matter \cite{Hehl:1997zh,hehl1997ahadronic,Puetzfeld:2007hr}.

\section{Equations of motion}
We are now in a position to derive the equations that describe the motion of the micro-structured body in the generic non-Riemannian geometry. Firstly, we derive the translational equations that describe how the momentum changes during the motion. To this end, we consider a 1-parameter family of curves $x^{\alpha}=x^{\alpha}(\lambda,\varepsilon)$ and the associated tetrads (or n-ads in general) $e_{\mu}{}^{b}(\lambda,\varepsilon)$ defined at each point of them, which are kept fixed under the $\varepsilon$-transport, that is, $\frac{\delta e_{\mu}{}^{b}}{\delta \varepsilon}=0$. We then need to extremize $I(\varepsilon)=\int_{\lambda_{1}}^{\lambda_{2}} L d\lambda$ subject to the fixed boundary conditions $x^{\mu}(\lambda_{i},\epsilon)=0$ for $i=1,2$. We have
\beq
\frac{dI}{d \varepsilon}=\int_{\lambda_{1}}^{\lambda_{2}}\left( \frac{\partial L}{\partial \upsilon^{\mu}}\frac{\delta \upsilon^{\mu}}{\delta \varepsilon}+\frac{\partial L}{\partial \dot{e}_{\mu}{}^{a}}\frac{\delta \dot{e}_{\mu}{}^{a}}{\delta \varepsilon}+\frac{\partial L}{\partial g_{\mu\nu}}\frac{\delta g_{\mu\nu}}{\delta \varepsilon}\right) d\lambda =0 
\eeq
Then using the identities (the proofs are given in the Appendix)
  \begin{equation}
    \frac{\delta \upsilon^{\mu}}{\delta \varepsilon}=\frac{\delta}{\delta \lambda}\left( \frac{\partial x^{\mu}}{\partial \varepsilon}\right)+ 2 S_{\alpha\beta}{}{}^{\mu}\upsilon^{\alpha}\frac{\partial x^{\beta}}{\partial \varepsilon} 
 \end{equation}
  \begin{equation}
    \frac{\delta }{\delta \varepsilon}\left( \frac{\delta e_{\mu}{}^{c}}{\delta \lambda}\right)=e_{\lambda}{}^{c}R^{\lambda}{}_{\mu\alpha\beta}\upsilon^{\alpha}\frac{\partial x^{\beta}}{\partial \varepsilon} 
 \end{equation}
 and the definition of non-metricity we derive the
translational equations of motion:
\beq
\boxed{\frac{D P_{\mu}}{d\lambda}=-2 S_{\mu\alpha\beta}\upsilon^{\alpha}P^{\beta}+\frac{1}{2}H^{\alpha\beta}R_{\alpha\beta\gamma\mu}\upsilon^{\gamma}-\frac{1}{2}Q_{\mu\alpha\beta}t^{\alpha\beta} }\label{eq1}
\eeq

The latter describe the evolution of the momentum during the motion. Expressed in this form, it is remarkable to observe how each geometric entity beautifully couples to the kinematical characteristics of matter; torsion couples to the canonical momentum (or more generally to the canonical energy-momentum tensor), curvature couples to hypermomentum and non-metricity couples to the metrical energy-momentum tensor. 

Moving on to the equations governing the 'rotational' hypermomentum part, we vary (\ref{action}) with respect to the material frame to readily arrive at
\beq
\frac{\partial L}{\partial e_{\mu}{}^{a}}-\frac{D}{d \lambda} \left(\frac{\partial L}{\partial \dot{e}_{\mu}{}^{a}} \right)=0 \label{framefe}
\eeq
We now use the identity (proof in the Appendix)
 \begin{equation}
 \frac{\partial L}{\partial \upsilon^{\nu}}\upsilon^{\mu}   -\frac{\partial L}{\partial e_{\mu}{}^{a}}e_{\nu}{}^{a}-\frac{\partial L}{\partial \dot{e}_{\mu}{}^{a}}\dot{e}_{\nu}{}^{a}-2 \frac{\partial L}{\partial g_{\mu\lambda}}g_{\nu\lambda}=0 
 \end{equation}
that comes from the diffeomorphism invariance of the particle action. Using the latter, and the definitions of momentum, hypermomentum and energy-momentum, we multiply (and contract) the field equations (\ref{framefe}) with $e_{\nu}{}^{a}$. After some trivial algebra we arrive at 
\beq
\frac{D H^{\nu\mu}}{d \lambda}=2 (P^{\nu}\upsilon^{\mu}-t^{\mu\nu})+H^{\beta\mu}Q_{\alpha\beta}{}{}^{\nu}\upsilon^{\alpha}  \label{eq2}
\eeq
which governs the evolution of the 'rotational' part of the particle. If we take the antisymmetric part of the latter we obtain the generalized evolution equation for the spin
\beq
\frac{D S^{\nu\mu}}{d \lambda}=2 P^{[\nu}\upsilon^{\mu]}+H^{\beta[\mu}Q_{\alpha\beta}{}{}^{\nu]}\upsilon^{\alpha}  
\eeq
in which we immediately observe that non-metricity acts as an additional torque acting on the test body.

Note that, if we bring the first index down, the last term on the right-hand side of the (\ref{eq2}) cancels and we obtain the more compact expression
\beq
\boxed{\frac{D H_{\nu}{}^{\mu}}{d \lambda}=2 (P_{\nu}\upsilon^{\mu}-t^{\mu}{}_{\nu}) }\label{Heq} 
\eeq

Performing a post-Riemannian expansion on the left-hand sides of (\ref{eq1}) and (\ref{eq2}), and bringing the momentum index up, after some trivial algebra, we may recast the system of equations in the expanded form
\begin{align}
\frac{\widetilde{D}P^{\nu}}{d\lambda}&=N_{\alpha\beta}{}{}^{\nu}(P^{\alpha}\upsilon^{\beta}-t^{\alpha\beta})+\frac{1}{2}H^{\alpha\beta}R_{\alpha\beta\gamma}{}{}{}^{\nu}\upsilon^{\gamma} \label{Peq} \\
\frac{\widetilde{D}H^{\nu\mu}}{d \lambda}&=2 (\upsilon^{\mu}P^{\nu}-t^{\mu\nu})+u^{\beta}(N^{\alpha\nu}{}{}_{\beta}H_{\alpha}{}^{\mu}-N^{\mu}{}_{\alpha\beta}H^{\nu\alpha})
\end{align}
The above system of equations, or equivalently (\ref{eq1}) together with (\ref{eq2}), fully describes the motion of a spinning-dilating-shearing  (microstructured) test body in generic non-Riemannian backgrounds with curvature, torsion and non-metricity. Our result generalizes the Mathisson-Papapetrou  equations for a spinning body \cite{Mathisson:1937zz,Papapetrou:1951pa} , including also the shear and dilation currents of matter. It also extends the result of \cite{Iosifidis:2023eom} in the sense that the canonical momentum $P_{\mu}$ and the energy-momentum tensor $t_{\mu\nu}$ are not specified at this point and are computed once the exact form of the Lagrangian is given. Furthermore, our result holds true not only for point particles but also for extended ones since it is valid to any multipole order. 

It is also worthwhile to note that the expressions for the equations of motion are formalistically the same regardless of the  choice of the Lagrangian. The Lagrangian does come into play though, when one needs to specify the canonical momentum and hypermomentum characterizing the moving material.

\section{Special Cases}

Let us now be more specific about the functional form of the particle Lagrangian in order to illustrate how our generic formalism reproduces certain results of the literature as special cases. To start with, we recall that when the mass is constant and the particle is structureless, moving in an external Riemannian background, the particle action reads 
\beq
S_{0}[g_{\mu\nu}, \upsilon^{\alpha}]=-m c\int \sqrt{-g_{\mu\nu}\dot{x}^{\mu}\dot{x}^{\nu}} d\lambda=-m c\int \sqrt{-g_{\mu\nu}\upsilon^{\mu}\upsilon^{\nu}} d\lambda
\eeq
where $m$ is the mass of the particle and $c$ is the speed of light we  shall set to unity ($c=1$) from now on. Note that there is no dependence on the material frame due to the fact that the particle is structureless. Of course, variations of the path give the geodesic equations of motion, $\frac{D P^{\alpha}}{d \lambda}=0$. Therefore, in order to have this Riemannian geometry-structureless particle limit, it is meaningful to assume that the full action can be broken into
\beq
I=S_{0}+S_{1}=-m \int \sqrt{-g_{\mu\nu}\upsilon^{\mu}\upsilon^{\nu}} d\lambda+\int L_{1}(\upsilon^{\mu}, e_{\mu}{}^{a}, \dot{e}_{\mu}{}^{a}) d \lambda 
\eeq
With this choice only the first part contributes to the energy-momentum tensor, which is given by
\beq
t^{\mu\nu}=m u^{\mu}u^{\nu}
\eeq
and only the second piece gives rise to hypermomentum. On the other hand, both parts contribute to the canonical momentum, since, obviously,
\beq
P_{\mu}=m u_{\mu}+\xi_{\mu}
\eeq
where $\xi_{\mu}=\frac{\partial L_{1}}{\partial \upsilon^{\mu}}$. For this choice of the Lagrangian, the system of equations describing the  trajectory of the microstructured body, takes the form
\begin{align}
\frac{\widetilde{D} P^{\nu}}{d \tau}&=\frac{1}{2}H^{\alpha\beta} u^{\gamma}R_{\alpha\beta\gamma}^{\;\;\;\;\;\;\nu}+N_{\beta\alpha}^{\;\;\;\;\nu}u^{\alpha}(P^{\beta}-m u^{\beta})
\\
\frac{\widetilde{D} H^{\nu\mu}}{d \tau}&=2 u^{\mu}(P^{\nu}-m u^{\nu})+u^{\beta}\Big( N^{\alpha\nu}{}{}_{\beta}H_{\alpha}{}^{\mu}-N^{\mu}{}_{\alpha\beta}H^{\nu\alpha} \Big)\label{Hdynamics}
\end{align}
which is in perfect agreement with \cite{Iosifidis:2023eom} (see also \cite{Obukhov:2015eqa}). Therefore, our Lagrangian description for the equations of motion is in accordance (and complementary) with the result obtained by integrating the conservation laws \cite{Iosifidis:2023eom} and the multipolar expansion method \cite{Obukhov:2015eqa}. 

\subsection{Pure Dilation}

Becoming more specific, let us consider an explicit form for $L_{1}$. The most simple example is that of a pure dilation particle, having the hypermomentum 
\beq
H^{\mu\nu}=\Delta g^{\mu\nu} \label{Hdil}
\eeq
where $\Delta$ is the dilation charge. It is not difficult to show that the extra Lagrangian piece that gives rise to such a contribution is $L_{1}=\Delta e^{\mu}{}_{a}\dot{e}_{\mu}{}^{a}$. Therefore, the full action reads
\beq
I= \int \Big( -m\sqrt{-g_{\mu\nu}\upsilon^{\mu}\upsilon^{\nu}}+\Delta e^{\mu}{}_{a}\dot{e}_{\mu}{}^{a}\Big) d\lambda \label{puredil}
\eeq
Variation with respect to the frame field gives $\dot{\Delta}=0 $, namely it implies that the dilation charge is constant, in perfect agreement with \cite{Iosifidis:2023eom}. The associated canonical momentum and energy-momentum corresponding to (\ref{puredil}) are easily found to be $P_{\mu}=m \upsilon_{\mu}$ and $t^{\mu\nu}=m \upsilon^{\mu}\upsilon^{\nu}$ respectively, that is, they do not receive contributions from $L_{1}$. With these at hand, we then vary the path, or directly substitute in (\ref{eq1}), to obtain the equations of the trajectory:
\beq
     \frac{d^{2}x^{\nu}}{d \tau^{2}}+\widetilde{\Gamma}^{\nu}_{\;\;\alpha\beta}u^{\alpha}u^{\beta} =\frac{1}{2 m}\Delta \upsilon_{\mu}\partial^{[\mu}Q^{\nu]}. \label{pathDilation}
\eeq
which show that there is an additional Lorentz-like force acting on the dilation charged particle, again in perfect agreement with \cite{Iosifidis:2023eom}. In this setting the particle may follow a geodesic only if the host spacetime has Weyl-integrable non-metricity vector, i.e. $Q_{\mu}=\partial_{\mu}\phi$ for some scalar $\phi$. In general, even in the simplest case of pure dilation, the trajectory is not a geodesic.

\section{Conserved Quantities}

Of course, it is an important aspect to be able to speak about quantities that are conserved during the motion. For this we first need to recall the definition of a generalized Killing vector field in Metric-Affine geometry. We have the following definition.
\begin{mydef*}
 Consider a generic Metric-Affine space. The vector field $\zeta$ is said to be a \underline{generalized Killing vector field} if it is simultaneously an isometry and an isoparallelism in the sense that its action leaves invariant the geometric background entities:
\begin{align}
\Lie_{\zeta}g_{\mu\nu}&=0 \;\;(isometry) \\
\Lie_{\zeta}\Gamma^{\lambda}{}_{\mu\nu}&=0 \;\; (isoparallelism)
\end{align}
\end{mydef*}

Suppose now that $\zeta^{\mu}$ is such a generalized Killing vector field. We then establish the following.
\begin{theorem*}
Let $\zeta^{\mu}$ be a generalized Killing vector field. Then the quantity
\beq
P_{\mu}\zeta^{\mu}+\frac{1}{2}H^{\mu\nu}\widetilde{\nabla}_{\mu}\zeta_{\nu}-H^{\mu\nu}N_{\mu\nu\kappa}\zeta^{\kappa}=const. \label{consquant}
\eeq
is a constant of motion\footnote{This result is in perfect agreement with \cite{Obukhov:2015eqa} where the multipole expansion method was used. It also agrees with \cite{Hojman:1978wk} when restricting to Riemann-Cartan spaces and vanishing dilation and shear charges.}, in the sense that it has a vanishing derivative along the trajectory, i.e.
\beq
    \frac{D}{d \lambda}\left(P_{\mu}\zeta^{\mu}+\frac{1}{2}H^{\mu\nu}\widetilde{\nabla}_{\mu}\zeta_{\nu}-\frac{1}{2}H^{\mu\nu}N_{\mu\nu\kappa}\zeta^{\kappa}\right) =0
\eeq
\end{theorem*}

\begin{proof}

 We start by contracting (\ref{Peq}) with $\zeta_{\nu}$. Using the Leibniz rule and also eq. (\ref{Heq}) we find
 \begin{gather}
\frac{D}{d \lambda}\Big( P_{\mu}\zeta^{\mu} \Big)=P^{\nu}\upsilon^{\mu}\widetilde{\nabla}_{\mu}\zeta_{\nu}+\frac{1}{2}H^{\alpha\beta}R_{\alpha\beta\gamma\mu}\upsilon^{\gamma}\zeta^{\mu}+\frac{1}{2}N^{\alpha}{}_{\beta\nu}\zeta^{\nu}\frac{D H_{\alpha}{}^{\beta}}{d \lambda} \label{Pzeta}
 \end{gather}
 where we have used the fact that $\zeta^{\nu}$ generates an isometry, namely $\widetilde{\nabla}_{(\mu}\zeta_{\nu)}=0$ is valid.
In a similar manner, contracting (\ref{Heq}) with $\widetilde{\nabla}_{\mu}\zeta^{\nu}$, using again Leibniz and the fact that $\zeta_{\mu}$ generates isometry, it follows that
\beq
2 P^{\nu}\upsilon^{\mu}\widetilde{\nabla}_{\mu}\zeta_{\nu}=\frac{D}{d \lambda}\Big( H^{\nu\mu}\widetilde{\nabla}_{\mu}\zeta_{\nu}\Big)-H_{\lambda}{}^{\mu}\upsilon^{\nu}\nabla_{\nu}\widetilde{\nabla}_{\mu}\zeta^{\nu} \label{intermediatestep}
\eeq

We now use the fact that $\zeta^{\mu}$ is also an isoparallelism.\footnote{Or equivalently, it generates an affine motion \cite{yano2020theory}.} This translates into the condition (see for instance \cite{yano2020theory})
\beq
\zeta^{\alpha}R^{\lambda}{}_{\mu\nu\alpha}=\nabla_{\nu}\nabla_{\mu}\zeta^{\lambda}+2 \nabla_{\nu}(S_{\mu\alpha}{}{}^{\lambda}\zeta^{\alpha})
\eeq
Expanding $\nabla_{\mu}\zeta^{\lambda}=\widetilde{\nabla}_{\mu}\zeta^{\lambda}+N^{\lambda}{}_{\alpha\mu}\zeta^{\alpha}$, the latter can be put in the form
\beq
\zeta^{\alpha}R^{\lambda}{}_{\mu\nu\alpha}=\nabla_{\nu}\widetilde{\nabla}_{\mu}\zeta^{\lambda}+2 \nabla_{\nu}(N^{\lambda}{}_{\mu\alpha}\zeta^{\alpha})
\eeq
and with this (\ref{intermediatestep}) takes the form
\beq
\frac{D}{d \lambda}\left( \frac{1}{2}H^{\mu\nu}\widetilde{\nabla}_{\mu}\zeta_{\nu}\right)=-P^{\nu}\upsilon^{\mu}\widetilde{\nabla}_{\mu}\zeta_{\nu}-\frac{1}{2}H^{\alpha\beta}R_{\alpha\beta\gamma\mu}\upsilon^{\gamma}\zeta^{\mu}+\frac{1}{2}H_{\alpha}{}^{\beta}\frac{D}{d \lambda}(N^{\alpha}{}_{\beta\nu}\zeta^{\nu})
\eeq
Adding this to (\ref{Pzeta}) the first two terms cancel each other and we observe that the remaining two on the right-hand side form an exact derivative, namely
\beq
 \frac{D}{d \lambda}\left(P_{\mu}\zeta^{\mu}+\frac{1}{2}H^{\mu\nu}\widetilde{\nabla}_{\mu}\zeta_{\nu}\right)=\frac{1}{2}\frac{D}{d \lambda}\Big(H^{\mu\nu}N_{\mu\nu\kappa}\zeta^{\kappa}\Big)
\eeq
Bringing all terms to the left-hand side we then complete the proof.

\end{proof}

\subsection{More Conserved Quantities}

If no restrictions on the properties of the matter are imposed, the combination (\ref{consquant}) is the only constant of motion that one gets quite generally. However, much like the only spin, Riemannian case further constants of motion can be obtained by imposing certain conditions on the properties of the body. Most notably, in the Riemannian case (i.e. for $Q_{\alpha\mu\nu}=0=S_{\alpha\mu\nu}$) and for matter which carries only the spin, namely $H_{\mu\nu}=H_{[\mu\nu]}=S_{\mu\nu}=-S_{\nu\mu}$,  $S_{\mu\nu}$ being the spin tensor, the spin magnitude $S_{\mu\nu}S^{\mu\nu}$ is conserved if one imposes either the Tulczyjew constraint $S_{\mu\nu}P^{\mu}=0$ or the  Pirani condition $S_{\mu\nu}P^{\mu}=0$. These are oftentimes refereed to as Spin Supplementary Conditions.\footnote{There are also other less popular forms of spin supplementary conditions, see for instance \cite{Costa:2014nta,Costa:2017kdr} for a discussion.}

In a similar manner to the aforementioned spin case, we can get extra constants of motion if we impose certain conditions on the (now) full hypermomentum tensor. To see this we contract (\ref{eq2}) with $H_{\mu\nu}$ and use the identity
\beq
H_{\mu\nu} \frac{D H^{\nu\mu}}{d \lambda}=\frac{1}{2}\frac{D}{d \lambda}\Big( H^{\nu\mu}H_{\mu\nu} \Big)+H^{\beta\mu}H_{\mu\nu}Q_{\alpha\beta}{}{}^{\nu}\upsilon^{\alpha}
\eeq
to find
\beq
\frac{D}{d \lambda}\Big( H^{\nu\mu}H_{\mu\nu} \Big)=4 (P^{\nu}\upsilon^{\mu}-t^{\mu\nu})H_{\mu\nu} \label{constr}
\eeq
Thus, if the hypermomentum tensor obeys $(P^{\nu}\upsilon^{\mu}-t^{\mu\nu})H_{\mu\nu}=0$, the magnitude $H^{\mu\nu}H_{\nu\mu}$ is conserved. For the only spin case this condition reduces to $S_{\mu\nu}P^{\mu}\upsilon^{\nu}$ being zero, which is valid if one imposes either of the Tulczyjew or Pirani constraints. In our Lagrangian description, however, much like the one in \cite{Bailey:1975fe}, such supplementary conditions do not need to be imposed.  Finally, let us observe that in the pure dilation case (recall eq. (\ref{Hdil})) the right-hand side vanishes identically and then eq. (\ref{constr}) implies $\Dot{\Delta}=0$ in accordance with the result of the previous section.

\section{Conclusions} 
We have developed the Lagrangian formulation that describes  the dynamics of extended bodies with microstructure (i.e. having non-zero hypermomentum), possessing all spin, dilation and shear charges, moving in generic Non-Riemannian backgrounds.  Consequently, using the least action principle, we  derived the equations of motion of the spinning-dilating-shearing multipole in generalized geometric backgrounds possessing  torsion, curvature and non-metricity. Our results generalize other Lagrangian descriptions of the Mathisson-Papapetrou equations of motion for spinning multipoles \cite{Bailey:1975fe}.

It should be noted that  our construction is fairly general. We only specified the fundamental variables upon which the particle Lagrangian  depends, without choosing a particular form. Apart from generality, the superiority of this development (canonical formalism) is that   it precisely defines the physical role of the variables involved; the canonical momentum for instance. It also  avoids any a priori constraints  between the spin tensor and the velocity or  momentum. 

Constraining the Lagrangian in a certain way, we made contact with existing literature on the equations of motion for hypermomentum charged particles moving in non-Riemannian backgrounds \cite{Obukhov:2015eqa,Iosifidis:2023eom}. 
Considering a generalized Killing vector field, and using the derived equations of motion, we explicitly obtained a quantity that is conserved along the trajectory. Using these results, it would then be interesting to study the motion of a micro-structured test body in non-Riemannian backgrounds. In addition, the results of this study can be used to construct singularity theorems in Metric-Affine geometries. Finally, changing slightly the functional form of our Lagrangian, it is possible to formulate the theory of a relativistic spinning-dilating-shearing top with internal degrees of freedom in non-Riemannian backgrounds. Some of these aspects are currently under investigation.

\section{Acknowledgments}

  This work was supported by the  Istituto Nazionale di Fisica Nucleare (INFN), Sezioni  di Napoli,  {\it Iniziative Specifiche} QGSKY. I would like to thank very much Friedrich Hehl and Yuri Obukhov for many useful comments and discussions.

\appendix

\section{Identities}

Let us prove here the important identities that were used in the various derivations. These identities are valid for generic Metric-Affine (i.e. non-Riemannian) geometries possessing all three geometric features of curvature torsion and non-metricity. Of course for vanishing torsion and non-metricity these give back some standard results of Riemannian geometry.

\begin{Identity}
  \begin{equation}
    \frac{\delta \upsilon^{\mu}}{\delta \varepsilon}=\frac{\delta}{\delta \lambda}\left( \frac{\partial x^{\mu}}{\partial \varepsilon}\right)+ 2 S_{\alpha\beta}{}{}^{\mu}\upsilon^{\alpha}\frac{\partial x^{\beta}}{\partial \varepsilon} \label{id1}
 \end{equation}
    
\end{Identity}

\begin{proof}
Firstly, recall that for $x^{\alpha}(\lambda,\varepsilon)$ we have that $\upsilon^{\alpha}:=\frac{\partial x^{\alpha}(\lambda,\varepsilon)}{\partial \lambda}$,  $\varepsilon^{\beta}:=\frac{\partial x^{\beta}(\lambda,\varepsilon)}{\partial \varepsilon}$ and also $\frac{\delta}{\delta \varepsilon}:=\varepsilon^{\alpha}\nabla_{\alpha}$ as well as  $\frac{\delta}{\delta \lambda}:=\upsilon^{\alpha}\nabla_{\alpha}$. Now, since $\upsilon^{\alpha}$ and $\varepsilon^{\beta}$ form a coordinate basis, they have a vanishing Lie bracket, namely
\beq
[\upsilon, \varepsilon]=0 \Leftrightarrow \varepsilon^{\alpha}\nabla_{\alpha}\upsilon^{\beta}=\upsilon^{\alpha}\nabla_{\alpha}
\varepsilon^{\beta}+2 S_{\alpha\gamma}{}{}^{\beta}\upsilon^{\alpha}\varepsilon^{\gamma} \label{Liebracket}
\eeq
From the latter and the above definitions, the identity (\ref{id1}) trivially follows.
    
\end{proof}

\begin{Identity}
  \begin{equation}
    \frac{\delta }{\delta \varepsilon}\left( \frac{\delta e_{\mu}{}^{c}}{\delta \lambda}\right)=e_{\lambda}{}^{c}R^{\lambda}{}_{\mu\alpha\beta}\upsilon^{\alpha}\frac{\partial x^{\beta}}{\partial \varepsilon} \label{id2}
 \end{equation}
    
\end{Identity}

\begin{proof}

    We start from the curvature identity
\beq
\nabla_{\alpha} \nabla_{\beta}e_{\mu}{}^{c}=\nabla_{\beta} \nabla_{\alpha}e_{\mu}{}^{c}-R^{\lambda}_{\;\;\;\mu\alpha\beta} e_{\lambda}{}^{c}+2 S_{\alpha\beta}^{\;\;\;\;\;\gamma}\nabla_{\gamma}e_{\mu}{}^{c}
\eeq
Contracting with $\varepsilon^{\alpha}\upsilon^{\beta}$ and using the Leibniz rule, identity (\ref{Liebracket}) along with 
\beq
\frac{\delta e_{\mu}{}^{b}}{\delta \varepsilon}=0 
\eeq
we readily arrive at (\ref{id2}).

\end{proof}

\begin{Identity}
  \begin{equation}
 \frac{\partial L}{\partial \upsilon^{\nu}}\upsilon^{\mu}   -\frac{\partial L}{\partial e_{\mu}{}^{a}}e_{\nu}{}^{a}-\frac{\partial L}{\partial \dot{e}_{\mu}{}^{a}}\dot{e}_{\nu}{}^{a}-2 \frac{\partial L}{\partial g_{\mu\lambda}}g_{\nu\lambda}=0 \label{id3}
 \end{equation}
    
\end{Identity}

\begin{proof}
    This identity follows as a consequence of the diffeomorphism invariance of the material action (\ref{action}). Let us prove. Consider the  coordinate transformation
    \beq
    x^{\mu}\mapsto x'^{\mu}=x^{\mu}+\xi^{\mu} \label{diffs}
    \eeq where $\xi^{\mu}$ is an infinitesimal. Then by differentiation we also have that $
\partial_{\mu}x'^{\nu}=\delta_{\mu}^{\nu}+\partial_{\mu}\xi^{\nu}
    $. Using this and the tensor transformation laws for each of $\upsilon^{\mu}, e_{\mu}{}^{a}, \dot{e}_{\mu}{}^{a}$ and $g_{\mu\nu}$ we readily find
    \beq
\delta \upsilon^{\nu}=-\xi^{\rho}\partial_{\rho}\upsilon^{\nu}+(\partial_{\mu}\xi^{\nu})\upsilon^{\mu}
    \eeq
    \beq
\delta e_{\mu}{}^{a}=-\xi^{\rho}\partial_{\rho}e_{\mu}{}^{a}-(\partial_{\mu}\xi^{\nu})e_{\nu}{}^{a}
    \eeq
  \beq
\delta \dot{e}_{\mu}{}^{a}=-\xi^{\rho}\partial_{\rho}\dot{e}_{\mu}{}^{a}-(\partial_{\mu}\xi^{\nu})\dot{e}_{\nu}{}^{a}
    \eeq
    \beq
\delta g_{\mu\nu}=-\xi^{\rho}\partial_{\rho}g_{\mu\nu}-2 g_{\lambda(\mu}\partial_{\nu)}\xi^{\lambda}
\eeq
Consequently the total action transforms as
\begin{gather}
\delta I=\int d\lambda \left[ \frac{\partial L}{\partial \upsilon^{\nu}}\delta\upsilon^{\mu}   -\frac{\partial L}{\partial e_{\mu}{}^{a}}\delta e_{\nu}{}^{a}-\frac{\partial L}{\partial \dot{e}_{\mu}{}^{a}}\delta \dot{e}_{\nu}{}^{a}-2 \frac{\partial L}{\partial g_{\mu\lambda}}\delta g_{\nu\lambda} \right]\nonumber \\
=\int d \lambda (\partial_{\mu}\xi^{\nu}) \left[ \frac{\partial L}{\partial \upsilon^{\nu}}\upsilon^{\mu}   -\frac{\partial L}{\partial e_{\mu}{}^{a}}e_{\nu}{}^{a}-\frac{\partial L}{\partial \dot{e}_{\mu}{}^{a}}\dot{e}_{\nu}{}^{a}-2 \frac{\partial L}{\partial g_{\mu\lambda}}g_{\nu\lambda} \right]
\nonumber \\
-\int d\lambda \xi^{\rho}\left[  \frac{\partial L}{\partial \upsilon^{\nu}}\partial_{\rho}\upsilon^{\nu}   +\frac{\partial L}{\partial e_{\mu}{}^{a}}\partial_{\rho}e_{\mu}{}^{a}+\frac{\partial L}{\partial \dot{e}_{\mu}{}^{a}}\partial_{\rho}\dot{e}_{\mu}{}^{a}+2 \frac{\partial L}{\partial g_{\mu\nu}}\partial_{\rho}g_{\mu\nu} \right] \label{deltaI}
\end{gather}
Now observe that the last parenthesis is the partial derivative of the Lagrangian since from the chain rule we have
\beq
\partial_{\rho}L\equiv  \frac{\partial L}{\partial \upsilon^{\nu}}\partial_{\rho}\upsilon^{\nu}   +\frac{\partial L}{\partial e_{\mu}{}^{a}}\partial_{\rho}e_{\mu}{}^{a}+\frac{\partial L}{\partial \dot{e}_{\mu}{}^{a}}\partial_{\rho}\dot{e}_{\mu}{}^{a}+2 \frac{\partial L}{\partial g_{\mu\nu}}\partial_{\rho}g_{\mu\nu}
\eeq
Furthermore from (\ref{diffs}) it is obvious that $\xi^{\rho}$ is proportional to $\frac{d x^{\rho}}{d \lambda}$ so there exists an infinitesimal constant $\epsilon$ such that $\xi^{\rho}=\epsilon \frac{d x^{\rho}}{d \lambda}$. Therefore the second part of the variation is a total derivative term
\begin{gather}
    \xi^{\rho}\left[  \frac{\partial L}{\partial \upsilon^{\nu}}\partial_{\rho}\upsilon^{\nu}   -\frac{\partial L}{\partial e_{\mu}{}^{a}}\partial_{\rho}e_{\mu}{}^{a}-\frac{\partial L}{\partial \dot{e}_{\mu}{}^{a}}\partial_{\rho}\dot{e}_{\mu}{}^{a}-2 \frac{\partial L}{\partial g_{\mu\nu}}\partial_{\rho}g_{\mu\nu} \right]=\nonumber \\
    =\epsilon \frac{d x^{\rho}}{d \lambda}\partial_{\rho}L=\frac{d}{d \lambda}(\epsilon L)
\end{gather}
and consequently  it can be dropped. Then since by construction the particle action is diffeomorphism invariant (i.e. $\delta I=0$) and given that $\partial_{\mu}\xi^{\nu}$ is arbitrary, it follows that the  parenthesis in the second line of (\ref{deltaI}) must be zero identically which proves the statement (\ref{id3}).
\end{proof}








\bibliographystyle{unsrt}
\bibliography{ref}

\end{document}